\documentclass[letterpaper, 10 pt, conference]{IEEEtran}

\usepackage{cite}

\usepackage[english]{babel}
\usepackage[cmex10]{amsmath}
\usepackage{amssymb,amsfonts}
\usepackage{amsthm}

\usepackage{graphicx}
\usepackage{etoolbox}

\usepackage{pgfplots}
\usepackage{pgfplotstable}

\pgfplotsset{compat=newest,
    /pgfplots/ybar legend/.style={
    /pgfplots/legend image code/.code={%
       \draw[##1,/tikz/.cd,yshift=-0.25em]
        (0cm,0cm) rectangle (3pt,0.8em);},
   },
}

\newtoggle{isreport}
\togglefalse{isreport}

\newtoggle{isarxiv}
\toggletrue{isarxiv}

\newtheorem{lemma}{Lemma}

\newtheorem{prop}[lemma]{Proposition}

\theoremstyle{remark}
\theoremstyle{theorem}
\theoremstyle{theorem}\newtheorem{defn}[lemma]{Definition}

\newcommand{\paren}[1]{\left(#1\right)}

\newcommand{\set}[1]{\left\{#1\right\}}
\newcommand{\abs}[1]{\left|#1\right|}

\newcommand{\normm}[1]{{\left\vert\kern-0.25ex\left\vert\kern-0.25ex\left\vert #1
    \right\vert\kern-0.25ex\right\vert\kern-0.25ex\right\vert}}

\DeclareMathOperator{\spann}{span}

\DeclareMathOperator{\diag}{diag}

\newcommand{\R}{\mathbb R}

\newcommand{\eps}{\epsilon}
\newcommand{\lam}{\lambda}

\newcommand{\ul}[1]{\underline{#1}}
\newcommand{\ol}[1]{\overline{#1}}

\newcommand{\goesto}{\rightarrow}

\newcommand{\bs}{\backslash}

\newcommand{\calB}{\mathcal{B}}
\newcommand{\calC}{\mathcal{C}}

\newcommand{\calE}{\mathcal{E}}
\newcommand{\calF}{\mathcal{F}}
\newcommand{\calG}{\mathcal{G}}

\newcommand{\calL}{\mathcal{L}}

\newcommand{\calN}{\mathcal{N}}

\newcommand{\calP}{\mathcal{P}}

\newcommand{\calT}{\mathcal{T}}

\newcommand{\bff}[1]{{\bf #1}}

\usepackage{color}
\usepackage{enumerate}

\definecolor{myred}{RGB}{202,0,32}
\definecolor{myorange}{RGB}{244,165,130}
\definecolor{myviolet}{RGB}{194,165,207}
\definecolor{mycyan}{RGB}{146,197,222}
\definecolor{myblue}{RGB}{5,113,176}
\definecolor{mygreen}{RGB}{127,191,123}
\definecolor{mytile}{RGB}{27,120,55}

\usepackage{tikz}
\newcommand*\circled[1]{\tikz[baseline=(char.base)]{
            \node[shape=circle,draw,inner sep=0.4pt] (char) {#1};}}

\IEEEoverridecommandlockouts

\title{
Less is More: Real-time Failure Localization in Power Systems
}

\author{Linqi Guo, Chen~Liang, Alessandro~Zocca, Steven H.~Low,~and~Adam~Wierman
\thanks{This work has been supported by Resnick Fellowship, Linde Institute Research Award, NWO Rubicon grant 680.50.1529., NSF grants through PFI:AIR-TT award 1602119, EPCN 1619352, CNS 1545096, CCF 1637598, ECCS 1619352, CNS 1518941, CPS 154471, AitF 1637598, ARPA-E grant through award {DE-AR0000699} (NODES) and GRID DATA, DTRA through grant HDTRA 1-15-1-0003 and Skoltech through collaboration agreement 1075-MRA.}
\thanks{The authors are with the Department of Computing and Mathematical Sciences, California Institute of Technology, Pasadena,
CA, 91125, USA. Email: \texttt{\{lguo, cliang2, azocca, slow, adamw\}@caltech.edu}.}}

\begin{document}

\maketitle
\thispagestyle{empty}
\pagestyle{empty}

\begin{abstract}
Cascading failures in power systems exhibit non-local propagation patterns which make the analysis and mitigation of failures difficult. In this work, we propose a distributed control framework inspired by the recently proposed concepts of \textit{unified controller} and network \textit{tree-partition} that offers strong guarantees in both the mitigation and localization of cascading failures in power systems. In this framework, the transmission network is partitioned into several control areas which are connected in a tree structure, and the  unified controller is adopted by generators or controllable loads for fast timescale disturbance response. After an initial failure, the proposed strategy always prevents successive failures from happening, and regulates the system to the desired steady state where the impact of initial failures are localized as much as possible. For extreme failures that cannot be localized, the proposed framework has a configurable design, that progressively involves and coordinates more control areas for failure mitigation and, as a last resort, imposes minimal load shedding. We compare the proposed control framework with Automatic Generation Control (AGC) on the IEEE 118-bus test system. Simulation results show that our novel framework greatly improves the system robustness in terms of the $N-1$ security standard, and localizes the impact of initial failures in majority of the load profiles that are examined. Moreover, the proposed framework incurs significantly less load loss, if any, compared to AGC, in all of our case studies.
\end{abstract}

\section{Introduction}\label{section:intro}
Cascading failures in power systems propagate non-locally, making their analysis and mitigation difficult. This fact is illustrated by the sequence of events leading to the 1996 Western US blackout summarized in Fig.~\ref{fig:blackout}, in which successive failures happened hundreds of kilometers away from each other (e.g.~from stage \circled{3} to stage \circled{4} and from stage \circled{7} to stage \circled{8}). Non-local propagation makes it particularly challenging to design distributed controllers that reliably prevent and mitigate cascades in power systems. In fact, such control is widely considered impossible, even when centralized coordination is available \cite{bienstock2007integer,hines2007controlling}.

Current industry practice for mitigating cascading failures mostly relies on simulation-based contingency analysis, which focuses on a small set of most likely initial failures \cite{baldick2008initial}. Moreover, the size of the contingency set which is tested (and thus the level of security guaranteed) is often constrained by computational power, undermining its effectiveness in view of the enormous number of components in power networks. After a blackout event, a detailed study typically leads to a redesign of such contingency sets, potentially together with physical network upgrades and revision of system management policies and regulations \cite{hines2007controlling}.

The limitations of current practice have motivated a large body of literature to study and characterize analytical properties of cascading failures in power systems. This literature can be roughly categorized as follows: (a) applying Monte-Carlo methods to analytical models that account for the steady state power redistribution using DC \cite{carreras2002critical,anghel2007stochastic,yan2015cascading,bernstein2014power} or AC \cite{nedic2006criticality,rios2002value,song2016dynamic} flow models; (b) studying pure topological models built upon simplifying assumptions on the propagation dynamics (e.g., failures propagate to adjacent lines with high probability) and inferring component failure propagation patterns from graph-theoretic properties \cite{brummitt2003cascade,kong2010failure,crucitti2004topological}; (c) investigating simplified or statistical cascading failure dynamics \cite{dobson2005probilistic,wang2012markov,rahnamay2014stochastic,hines2017cascading}. 
In all these approaches, the non-local failure propagation often creates significant challenges when trying to make general inferences about failure patterns.

\begin{figure}[t]
\centering
{
\iftoggle{isarxiv}{
\includegraphics[width=0.38\textwidth]{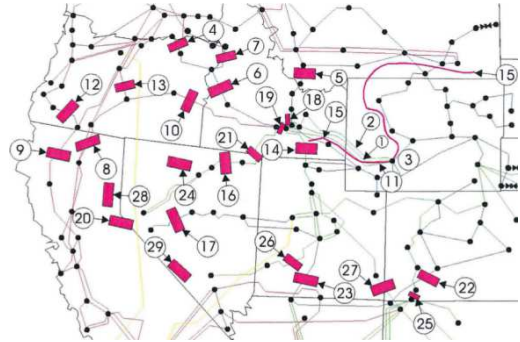}
}
{
\includegraphics[width=0.38\textwidth]{figs/us_blackout}
}
}
\caption{The sequence of events, indexed by the circled numbers, that lead to the Western US blackout in 1996. Adopted from \cite{hines2017cascading}.}
\label{fig:blackout}
\end{figure}

A new approach to address such challenges has emerged in recent years, aiming to improve the system resilience by partitioning the power network into smaller regions and localizing impacts of line failures within each region \cite{guo2017monotonicity,guo2018failure}. This approach is based on the \emph{tree-partition} of power grids (see Section \ref{section:tree_partition} for its definition and properties) and yields many interesting, sometimes counter-intuitive, insights for the planning and management of power systems. For instance, the case studies in \cite{guo2018failure} show that, by properly reducing system redundancy, one can improve system robustness against component failures -- the existence of redundancy turns out to play a prominent role in the non-local nature of failure propagation in power systems.

Unfortunately, this literature has yet to yield a fully satisfactory solution for mitigating and localizing failures, due to two main challenges. First, reducing redundancy as suggested in \cite{guo2017monotonicity,guo2018failure} may lead to single-point vulnerabilities, whose failure has a global impact on the whole system and can potentially cause significant load loss. Second, information on unfolding cascading failures are not fed back into relevant controllers that could adjust the network topology (and in particular its tree-partition). Therefore, after an initial failure is triggered, the strategy described in \cite{guo2018failure} guarantees that any successive failure will occur in the same region as the initial failure, but does not prevent or stop successive failures from happening in the first place. To overcome these drawbacks, there is need for new control designs that can ``close the loop'' and respond actively and promptly to different failures.

\textbf{Contributions of this paper:} \emph{We propose a distributed control strategy that operates on the frequency regulation timescale and offers proveable failure mitigation properties and localization guarantees.} To the best of our knowledge, this control paradigm is the first to leverage results from the frequency regulation literature in the context of cascading failures, bringing new perspectives and insights to both literatures. Our control scheme guarantees that failures do not propagate whenever there is a feasible way to avoid it (see Section \ref{section:basic} on the rigorous definition of such feasibility), and the impact of failures are localized as much as possible in a manner configurable by the system operator. 

We introduce the main idea of our control design in Section \ref{section:control_strategy}, whose failure mitigation and localization guarantees are established by the technical results in Section \ref{section:non_critical} and \ref{section:critical}. The key piece of our control builds upon the so-called Unified Controller (UC), a noval design approach to  frequency regulation \cite{zhao2014design,mallada2017optimal,zhao2016unified,zhao2018distributed}. Our design revolves around the new and powerful properties that emerge when the regions that UC manages form a tree-partition. More specifically, in Section \ref{section:non_critical}, we characterize how UC responds to an initial failure when it operates over a tree-partition, and prove that a non-critical failure is always mitigated and localized. Later, in Section \ref{section:critical}, we discuss how the tree-partition enables the system operator to explicitly specify the unfolding pattern of critical failures, and prove that UC can be extended to detect such scenarios as part of its normal operation.

In order to establish these results, we make use of the correspondence between swing/power flow dynamics and primal-dual algorithms to solve certain optimal dispatch problems, as developed in \cite{zhao2014design,li2016connecting,zhao2016unified}. Further, we prove new results on the UC optimization problem using novel algebraic characterizations of DC power flow equations derived from graph-theoretic properties of tree-partitions. Lastly, we apply the classical results from convex analysis and optimization theory to show that critical failures can always be detected by UC in a distributed fashion. 

In Section \ref{section:case_studies}, we compare the proposed control strategy with classical Automatic Generation Control (AGC) using the IEEE 118-bus test system. We demonstrate that by switching off only a small subset of transmission lines and adopting UC as the fast timescale controller, one can significantly improve the system robustness to failures in terms of the $N-1$ security standard. Moreover, in a majority of the load profiles that are examined, our control strategy further localizes the impact of initial failures to the regions where they occur, leaving the operating points of all other control areas unchanged. Lastly, we highlight that when load shedding is inevitable, the proposed framework incurs significantly less load loss compared to AGC, in all of our case studies.

\section{Preliminaries}\label{section:preliminaries}
In this section, we present our power system model, describe the failure propagation process thus derived, and discuss how they are related to previous models in literature.

\subsection{Power Grid Model and Fast Timescale Dynamics}\label{section:model}
We consider a power transmission network described by the graph $\calG=(\calN,\calE)$, where $\calN=\set{1,\ldots, \abs{\calN}}$ is the set of buses and $\calE\subset\calN \times \calN$ is the set of transmission lines. The terms bus/node and line/edge will be used interchangeably in the rest of the paper. An edge in $\calE$  between nodes $i$ and $j$ is denoted either as $e$ or $(i,j)$. We assign a fixed but arbitrary orientation to the edges in $\calE$, so that if $(i,j)\in\calE$ then $(j,i)\notin\calE$. Together with the variables associated with buses and transmission lines from Table \ref{table:notations}, the linearized swing and power flow dynamics can be written as
\begin{subequations}\label{eqn:swing_and_network_dynamics}
\begin{IEEEeqnarray}{rCll}
  M_j\dot{\omega}_j&=& r_j - d_j - D_j\omega_j - \sum_{e\in\calE}C_{je}f_e,&\quad j\in\calN \label{eqn:swing_dynamics}\\
  \dot{f}_{ij}&=&B_{ij}(\omega_i-\omega_j),\quad& (i,j)\in\calE \label{eqn:network_flow_dynamics}
\end{IEEEeqnarray}
\end{subequations}
We refer the readers to \cite{zhao2014design,zhao2016unified} for more detailed justification and derivation of this model. 

\begin{table}[t]
\vspace{.2cm}
\def\arraystretch{1.1}
\caption{Variables associated with buses and transmission lines.}\label{table:notations}
\begin{tabular}{m{2.7cm}|m{5cm}}
\hline
$\theta:=(\theta_j, j\in\calN)$ & bus voltage angle deviations\\
\hline
$\omega:=(\omega_j, j\in\calN)$ & bus frequency deviations \\
\hline
$r:=(r_j, j\in\calN)$ & injection deviations\\
\hline
$d:=(d_j, j\in\calN)$ & mechanical power injection adjustment for generator buses; controllable load adjustment for load buses\\
\hline
$\ol{d}_j, \ul{d}_j, j\in\calN$ & upper and lower limits for the adjustable injection $d_j$\\
\hline
$D_j\omega_j, j\in\calN$ & frequency sensitive generator/load dynamics\\
\hline
$f:=(f_e, e\in\calE)$ & branch flow deviations\\
\hline
$\ol{f}_e, \ul{f}_e, e\in\calE$ & upper and lower limits for branch flow deviations\\
\hline
$C\in \R^{\abs{\calN}\times \abs{\calE}}$ & incidence matrix of $\calG$: $C_{je}=1$ if $j$ is the source of $e$, $C_{je}=-1$ if $j$ is the destination of $e$, and $C_{je}=0$ otherwise\\
\hline
$B:=\diag(B_e, e\in\calE)$ & branch flow linearization coefficients that depend on nominal state voltage magnitudes and reference phase angles \\
\hline
\end{tabular}
\vspace{-.2cm}
\end{table}

\begin{defn}
A state $x^*:=(\omega^*, d^*, f^*)\in\R^{2\abs{\calN}+\abs{\calE}}$ is said to be an \textbf{\emph{equilibrium}} of \eqref{eqn:swing_and_network_dynamics} if the right hand sides of \eqref{eqn:swing_and_network_dynamics} are zero at $x^*$. 
\end{defn}

We emphasize that the $d_j$'s in \eqref{eqn:swing_dynamics} usually depend on the system states and may evolve by themselves in accordance to certain controller specific dynamics. The equilibrium defined above refers to the \emph{closed-loop} equilibrium. It is thus possible to engineer the equilibrium of \eqref{eqn:swing_and_network_dynamics} by adopting a different controller design for $d_j$, which in turn impacts how failures propagate in the system.

\subsection{Failure Occurrence and Propagation}\label{section:propagation}

In full generality, the control strategy that we introduce later applies to both generator failures and line failures. However, to simplify the presentation, in this paper we focus only on line failures as the generalization to bus failures is straightforward.\footnote{Our results readily apply to cases where the failure of a generator or substation can be emulated by the simultaneous failures of all the transmission lines connected to the corresponding bus.}

We describe the cascading failure process by keeping track of the set of failed lines at each stage, which are naturally nested and expanding as the stages progress. Overloaded lines are tripped at slower timescales than the dynamics \eqref{eqn:swing_and_network_dynamics}; the cascade stages reflect this fact and, indeed, at each stage we \emph{assume} the system reaches the new steady state equilibrium. The crux of our failure propagation model lies in the interplay between such slow timescale line tripping process and the fast timescale dynamics on system transient behavior described by \eqref{eqn:swing_and_network_dynamics}, as illustrated in Fig.~\ref{fig:model_interplay}.

\begin{figure}[t]
\vspace{.1cm}
\centering
{
\iftoggle{isarxiv}{
\begin{tikzpicture}
\begin{axis}[
        ymax    = 50.1,
        ymin    = 49.85,
        xmax    = 25,
        xmin    = 0,
        width   = 8 cm,
        height  = 3.4 cm,
        xlabel style = {font = \small},
        ylabel style = {font = \small},
        xticklabel style = {font = \small},
        yticklabel style = {font = \small},
    	xlabel = {Time},
    	ylabel = {Frequency (Hz)},
	]
    
	\pgfplotstableread{illustration_1.txt}\datatable
	\addplot +[
	mark = none,
        color       = myblue,
        line width  = 1pt,
        line cap    = round,
        line join   = round,
    ] table {\datatable};
    
	\pgfplotstableread{illustration_2.txt}\datatable
	\addplot +[
	mark = none,
        color       = myblue,
        line width  = 1pt,
        line cap    = round,
        line join   = round,
    ] table {\datatable};
    
    	\pgfplotstableread{illustration_3.txt}\datatable
	\addplot +[
	mark = none,
        color       = myblue,
        line width  = 1pt,
        line cap    = round,
        line join   = round,
    ] table {\datatable};
    
    	\pgfplotstableread{illustration_4.txt}\datatable
	\addplot +[
	mark = none,
        color       = myblue,
        dashed,
        line width  = 1pt,
        line cap    = round,
        line join   = round,
    ] table {\datatable};
    
    	\pgfplotstableread{illustration_5.txt}\datatable
	\addplot +[
	mark = none,
        color       = myblue,
        line width  = 1pt,
        line cap    = round,
        line join   = round,
    ] table {\datatable};

    \addplot + [mark=none, color = myviolet, dashed, line width = 1pt] coordinates {(5,49.85) (5,50.1)};
    \addplot + [mark=none, color = myviolet, dashed, line width = 1pt] coordinates {(10,49.85) (10,50.1)};
    \addplot + [mark=none, color = myviolet, dashed, line width = 1pt] coordinates {(15,49.85) (15,50.1)};
    \addplot + [mark=none, color = myviolet, dashed, line width = 1pt] coordinates {(20,49.85) (20,50.1)};
    	\begin{scope}[
  	every pin edge/.style={latex-},
  	pin distance=0.4cm
	]
	\node[coordinate,pin=  60:{\small{initial failure}}] at (axis cs:5,50) {};
	\end{scope}
    \node[] at (axis cs:2.5, 49.9){\small{stage 0}};
    \node[] at (axis cs:7.5, 49.9){\small{stage 1}};
    \node[] at (axis cs:12.5, 49.9){\small{stage 2}};
    \node[] at (axis cs:17.5, 49.9){\small{...}};
    \node[] at (axis cs:22.5, 49.9){\small{stage $N$}};
\end{axis}
\end{tikzpicture}
}{
\begin{tikzpicture}
\begin{axis}[
        ymax    = 50.1,
        ymin    = 49.85,
        xmax    = 25,
        xmin    = 0,
        width   = 8 cm,
        height  = 3.4 cm,
        xlabel style = {font = \small},
        ylabel style = {font = \small},
        xticklabel style = {font = \small},
        yticklabel style = {font = \small},
    	xlabel = {Time},
    	ylabel = {Frequency (Hz)},
	]
    
	\pgfplotstableread{illustration_1.txt}\datatable
	\addplot +[
	mark = none,
        color       = myblue,
        line width  = 1pt,
        line cap    = round,
        line join   = round,
    ] table {\datatable};
    
	\pgfplotstableread{illustration_2.txt}\datatable
	\addplot +[
	mark = none,
        color       = myblue,
        line width  = 1pt,
        line cap    = round,
        line join   = round,
    ] table {\datatable};
    
    	\pgfplotstableread{illustration_3.txt}\datatable
	\addplot +[
	mark = none,
        color       = myblue,
        line width  = 1pt,
        line cap    = round,
        line join   = round,
    ] table {\datatable};
    
    	\pgfplotstableread{illustration_4.txt}\datatable
	\addplot +[
	mark = none,
        color       = myblue,
        dashed,
        line width  = 1pt,
        line cap    = round,
        line join   = round,
    ] table {\datatable};
    
    	\pgfplotstableread{illustration_5.txt}\datatable
	\addplot +[
	mark = none,
        color       = myblue,
        line width  = 1pt,
        line cap    = round,
        line join   = round,
    ] table {\datatable};

    \addplot + [mark=none, color = myviolet, dashed, line width = 1pt] coordinates {(5,49.85) (5,50.1)};
    \addplot + [mark=none, color = myviolet, dashed, line width = 1pt] coordinates {(10,49.85) (10,50.1)};
    \addplot + [mark=none, color = myviolet, dashed, line width = 1pt] coordinates {(15,49.85) (15,50.1)};
    \addplot + [mark=none, color = myviolet, dashed, line width = 1pt] coordinates {(20,49.85) (20,50.1)};
    	\begin{scope}[
  	every pin edge/.style={latex-},
  	pin distance=0.4cm
	]
	\node[coordinate,pin=  60:{\small{initial failure}}] at (axis cs:5,50) {};
	\end{scope}
    \node[] at (axis cs:2.5, 49.9){\small{stage 0}};
    \node[] at (axis cs:7.5, 49.9){\small{stage 1}};
    \node[] at (axis cs:12.5, 49.9){\small{stage 2}};
    \node[] at (axis cs:17.5, 49.9){\small{...}};
    \node[] at (axis cs:22.5, 49.9){\small{stage $N$}};
\end{axis}
\end{tikzpicture}}}
\caption{An illustration of the failure propagation model.}
\label{fig:model_interplay}
\vspace{-.3cm}
\end{figure}
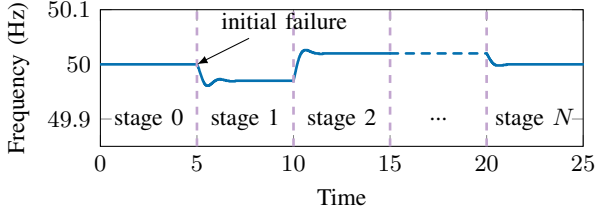

More specifically, each stage $n\in \set{1,2,\ldots, N}$ corresponds to a topology $\calG(n):=(\calN, \calE \bs \calB(n))$ in the failure unfolding process, where $\calB(n)$ is the set of all tripped lines at stage $n$. Within stage $n$, the system evolves according to the dynamics \eqref{eqn:swing_and_network_dynamics} on the topology $\calG(n)$, and converges to an equilibrium point $x^*(n)=(\omega^*(n), d^*(n), f^*(n))$ that depends on $\calG(n)$. 
If all the branch flows $f^*(n)$ are below the corresponding line ratings at equilibrium, then $x^*(n)$ is a secure operating point and the cascade stops. Otherwise, let $\calF(n)$ be the subset of lines whose branch flow exceeds the corresponding line ratings. The lines in $\calF(n)$ operate above their safety limits in steady state, so by the end of stage $n$ they are overheated and tripped, i.e., $\calB(n+1)=\calB(n)\cup \calF(n)$. Line overloads during the transient phase before the system converges to $x^*(n)$ are considered to be tolerable because the transient dynamics in \eqref{eqn:swing_and_network_dynamics} are not long enough to overheat a line \cite{zhao2016unified} (lasting only seconds to a few minutes). This process repeats for stage $n+1$ and so on.

\subsection{Recovering Previous Models}\label{section:previous_model}
Our failure propagation model brings new perspectives to the commonly studied models in literature, and reveals interesting insights on how certain  limitations from previous work can be circumvented. In particular, the extra freedom in choosing $d_j$ in the fast timescale dynamics \eqref{eqn:swing_and_network_dynamics} allows us to design and improve how the system reacts to line failures; thus achieving failure mitigation objectives directly using the well-known analytical tools from frequency regulation literature. 

As a first example, we show that adopting the classical droop control \cite{bergen2009power} in our framework recovers cascading failure models from previous literature such as \cite{soltan2015analysis, yan2015cascading, guo2017monotonicity, guo2018failure}.  Indeed, as shown in \cite{zhao2016unified}, the closed-loop equilibrium of \eqref{eqn:swing_and_network_dynamics} under droop control is the unique\footnote{Such uniqueness is up to a constant shift of all phase angles $\theta$. See \cite{zhao2016unified}.} optimal solution to the following optimization:
\begin{subequations}\label{eqn:droop_olc}
\begin{IEEEeqnarray}{ll}
\min_{\omega, d, f, \theta} \quad & \sum_{j\in\calN}\frac{d_j^2}{2K_j}+\frac{D_jw_j^2}{2}\label{eqn:droop_obj}\\
 \hspace{.2cm}\text{s.t.}& r - d - D\omega = Cf\label{eqn:droop_balance}\\
	& f - BC^T \theta = 0\label{eqn:droop_branch}\\
	&\ul{d}_j\le d \le \ol{d}_j, \quad j\in\calN, \label{eqn:droop_control_limits}
\end{IEEEeqnarray}
\end{subequations}
where $K_j$'s are the generators' participation factors \cite{bergen2009power}. By plugging \eqref{eqn:droop_branch} into \eqref{eqn:droop_balance}, it is routine to check that any feasible point $x=(\omega, d, f, \theta)$ of \eqref{eqn:droop_olc} satisifes $\sum_j r_j=\sum_j (d_j + D_j\omega_j)$. 
As a result, the Cauchy-Schwarz inequality implies that
\begin{IEEEeqnarray*}{rCl}
\Bigg( \sum_{j\in\calN} r_j \Bigg)^2&=& \Bigg[ \sum_{j\in\calN} \paren{d_j+D_j\omega_j}\Bigg]^2\\
&\le & \sum_{j\in\calN}\paren{\frac{d_j^2}{2K_j}+\frac{D_j\omega_j^2}{2}}\sum_{j\in\calN}\paren{2K_j+2D_j},
\end{IEEEeqnarray*}
and equality holds if and only if 
\begin{equation}\label{eqn:droop_optimal_point}
d_j=\frac{K_j}{\sum_j \paren{K_j+D_j}}\sum_j r_j, \quad \omega_j=\frac{\sum_j r_j}{\sum_j \paren{K_j+D_j}}.
\end{equation}
Therefore, if the control limits \eqref{eqn:droop_control_limits} are not active, \eqref{eqn:droop_optimal_point} is always satisfied at the optimal point $x^*=(\omega^*, d^*, f^*, \theta^*)$.

Now consider a line $e$ being tripped from the transmission network $\calG$, and for simplicity assume the control limits \eqref{eqn:droop_control_limits} are not active. If $e$ is a bridge\footnote{A line $e$ is said to be a \textit{bridge} for $\calG$ if it is a cut-edge for $\calG$, i.e., if the removal of $e$ from $\calG$ disconnects $\calG$ into two components, usually referred to as islands in power system literature. See \cite{bondy1976graph} for its rigorous definition.}, the tripping of $e$ results in two islands of $\calG$, say $\calN_1$ and $\calN_2$, and two optimization problems \eqref{eqn:droop_olc} corresponding to $\calN_1$ and $\calN_2$ respectively. For $l=1,2$, $\sum_{j\in\calN_l} r_j$ represents the total power imbalance in $\calN_l$, and therefore \eqref{eqn:droop_optimal_point} implies that droop control adjusts the system injection so that the power imbalance is distributed to all generators proportional to their participation factors in both $\calN_1$ and $\calN_2$. If $e=(i,j)$ is not a bridge, denoting the original flow on $e$ before it is tripped as $f_e$, then $r_i=f_e$, $r_j=-f_e$ and $r_k=0$ otherwise. As a result, we have $\sum_{j\in\calN}r_j = 0$ in this case and thus \eqref{eqn:droop_optimal_point} implies the system operating point remains unchanged. This control recovers exactly the failure propagation dynamics in \cite{soltan2015analysis, yan2015cascading, guo2017monotonicity, guo2018failure}. Moreover, one can show that this still holds when \eqref{eqn:droop_control_limits} is active with a more involved analysis on the KKT conditions of \eqref{eqn:droop_olc}.

We thus see that this droop control mechanism underlies some of the previous results in the literature on cascading failures in power systems. In paricular, this suggests that, by using a different control design for $d_j$, we can obtain different and potentially better system behaviors after a line failure. For instance, it is shown in \cite{guo2018failure} that bridge failures under droop control have a global impact, while (as we outline in Section \ref{section:non_critical}) the impact of bridge failures can in fact be localized using UC. Our new proposed control strategy leverages precisely this extra freedom in chooseing the $d_j$'s to offer stronger guarantees in both failure mitigation and localization compared to previous work \cite{guo2018failure}.

\section{Tree-partitions and the Unified Controller}\label{section:basic}
The tree-partition and UC have emerged recently as tools to improve power system robustness \cite{guo2017monotonicity, guo2018failure, zhao2016unified}. These concepts have been investigated separately in the literature as they operate at different timescales and aim to solve different problems. Our model brings them together into a novel framework, which allows us to obtain new results combining their strengths and, at the same time, provides new insights for both. In this section, we review these concepts and explain how they come together in this work.

\subsection{The Tree-partition} \label{section:tree_partition}
\begin{figure}
\centering
\iftoggle{isarxiv}{
\includegraphics[width=0.325\textwidth]{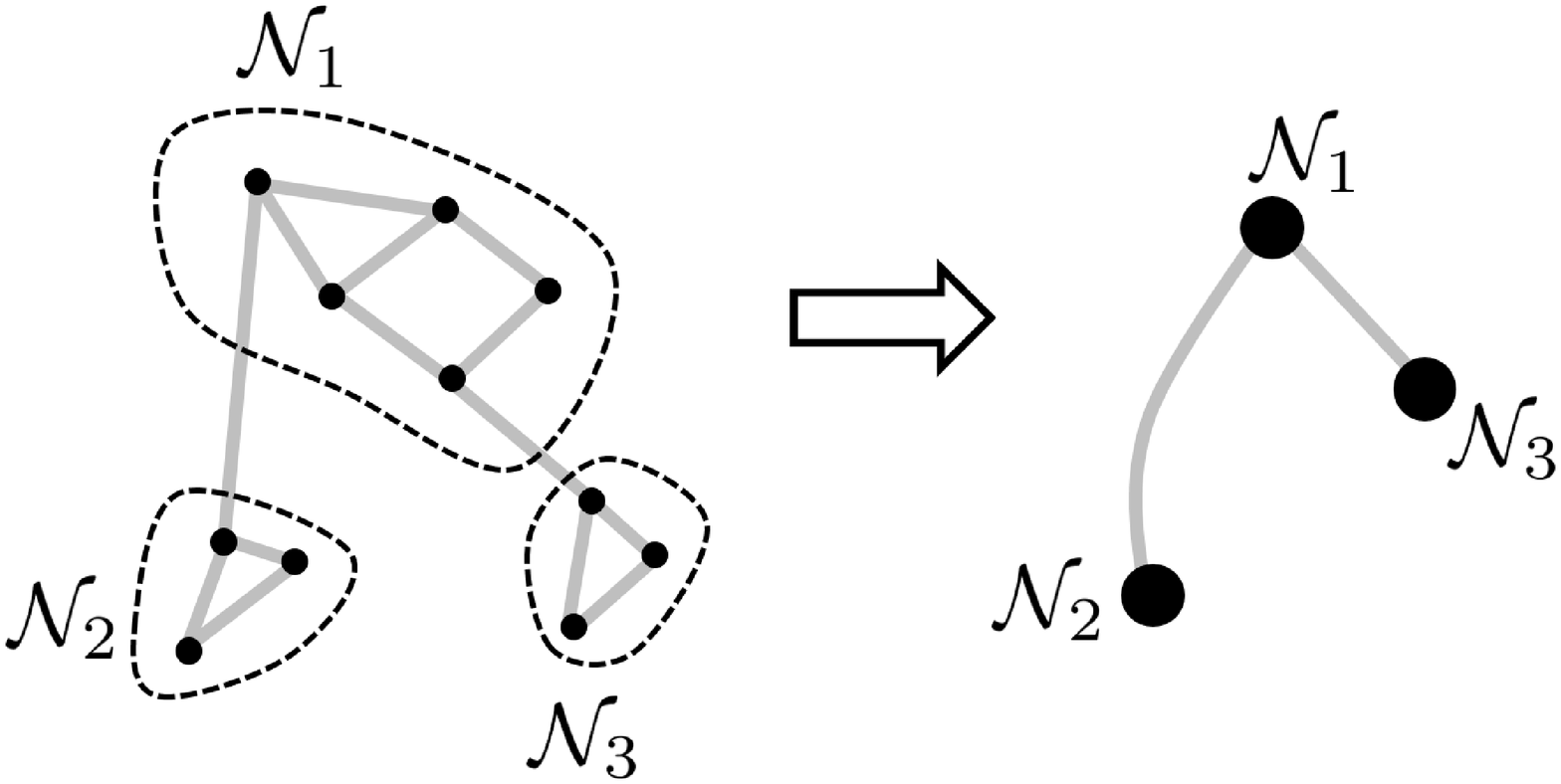}
}{
\includegraphics[width=0.325\textwidth]{figs/tree_partition.png}
}
\caption{The construction of $\calG_\calP$ from $\calP$.}\label{fig:tree_partition}
\end{figure}

Given a power system $\calG=(\calN, \calE)$, a collection $\calP=\set{\calN_1,\calN_2,\cdots, \calN_k}$ of subsets of $\calN$ is said to form a \emph{partition} of $\calG$ if $\calN_i\cap\calN_j=\emptyset$ for $i\neq j$ and $\cup_{i=1}^k \calN_i=\calN$.  For each partition $\calP=\set{\calN_1,\calN_2,\cdots, \calN_k}$, we can define a reduced multi-graph $\calG_\calP$ from $\calG$ as follows (see Fig.~\ref{fig:tree_partition}). The node set of $\calG_\calP$ is in one-to-one correspondence with $\set{\calN_1,\calN_2,\cdots, \calN_k}$, in the sense that we collapse each subset $\calN_i$ into a ``super node'' of $\calG_\calP$. We then add an undirected edge connecting the super nodes $\calN_i$ and $\calN_j$ for each pair of nodes $v, w\in\calN$ if $v\in\calN_i$, $w\in\calN_j$ and they are connected in $\calG$, i.e. $(v,w)\in\calE$ or $(w,v)\in\calE$. Note that multiple edges are added when multiple pairs of such $v,w$ exist.

\begin{defn}\label{defn:tree_partition}
A partition $\calP^{\text{\emph{tree}}}=\set{\calN_1,\calN_2,\cdots,\calN_k}$ of $\calG$ is said to form a \textbf{\emph{tree-partition}} if the reduced multi-graph $\calG_{\calP^{\text{\emph{tree}}}}$ is a tree. In this case, the subsets $\calN_l$ are referred to as tree-parition \textbf{\emph{regions}} and the edges of $\calG$  whose endpoints belong to different regions are called \textbf{\emph{bridges}}.
\end{defn}

It is shown in \cite{guo2018failure} that each graph $\calG$ has a unique irreducible tree-partition, which can be computed in linear time, and for this tree-partition, the concept of bridges defined above coincides with that from classical graph theory literature, such as \cite{bondy1976graph}. The tree-partition of a transmission network encodes rich information about how failures propagate (under droop control, as we discussed in Section \ref{section:previous_model}).

\subsection{The Unified Controller (UC)}\label{section:UC}
UC is a control approach recently proposed in the frequency regulation literature \cite{zhao2014design, mallada2017optimal, zhao2016unified,zhao2018distributed}. Compared to classical droop control or Automatic Generation Control (AGC) \cite{bergen2009power}, UC aims to achieve primary frequency control, secondary frequency control, and congestion management simultaneously at the frequency control timescale.

The key feature of UC that we use here is that the closed-loop equilibrium of \eqref{eqn:swing_and_network_dynamics} under UC solves the following optimization:
\begin{subequations}\label{eqn:uc_olc}
\begin{IEEEeqnarray}{ll}
\min_{f, d, \theta} \quad& \sum_{j\in\calN}c_j(d_j) \label{eqn:uc_obj}\\
\hspace{.1cm}\text{s.t.} & r - d - Cf = 0 \label{eqn:uc_balance}\\
& f = BC^T\theta  \label{eqn:dcflow}\\
& ECf = 0\label{eqn:uc_ace}\\
& \ul{f}_{e}\le f_{e}\le \ol{f}_{e}, \quad e\in\calE \label{eqn:line_limit}\\
& \ul{d}_{j}\le d_j\le \ol{d}_j, \quad j \in\calN, \label{eqn:control_limit}
\end{IEEEeqnarray}
\end{subequations}
where $c_j(\cdot)$'s are associated cost functions that penalize deviations from last optimal dispatch point (and hence attain minimum at $0$), \eqref{eqn:uc_balance} guarantees power balance at each bus,  \eqref{eqn:dcflow} is the DC power flow equation, \eqref{eqn:uc_ace} enforces zero area control error \cite{bergen2009power}, \eqref{eqn:line_limit} and \eqref{eqn:control_limit} are the flow and control limits. The matrix $E$ encodes control area information as follows: Given a partition $\calP^{\text{UC}}= \set{\calN_1,\calN_2,\cdots, \calN_k}$ of $\calG$ that specifies the control areas in secondary frequency control, $E\in\set{0,1}^{|\calP^{\text{UC}}|\times \abs{\calN}}$ is defined by $E_{l,j}=1$ if bus $j$ is in region $\calN_l$ and $E_{l,j}=0$ otherwise. An edge $e\in\calE$ is called a \emph{tie-line} if its endpoints belong to different regions in $\calP^{\text{UC}}$ \cite{bergen2009power,zhao2016unified}. 
As a result, the $l$-th row of $ECf=0$ ensures that the branch flow deviations on the tie-lines connected to $\calN_l$ sum to zero.

UC is designed so that its controller dynamics combined with the system dynamics \eqref{eqn:swing_and_network_dynamics} form a variant of projected primal-dual algorithms to solve \eqref{eqn:uc_olc}. It is shown in \cite{zhao2016unified,zhao2018distributed} that when the optimization problem \eqref{eqn:uc_olc} is feasible, under mild assumptions UC is globally stable and converges to the optimal point of \eqref{eqn:uc_olc}. 
This optimal point is unique (up to a constant shift of $\theta$) if the cost functions $c_j(\cdot)$ are strictly convex. We refer readers to \cite{zhao2016unified,zhao2018distributed} for its exact controller design and analysis.

\subsection{Connecting UC to Tree-partition}\label{section:connecting_uc_to_tree}
In the previous subsections we mention two distinct partitions of a power network: the tree-partition $\calP^{\text{tree}}$ and the control area partition $\calP^{\text{UC}}$. In general, $\calP^{\text{tree}}$ and $\calP^{\text{UC}}$ can be different. However, when they do coincide, the underlying power grid inherits analytical properties from both tree-partition and UC, making the system particularly robust against failures. Our proposed control strategy leverages this connection, as we present in more detail in Section \ref{section:control_strategy}, and we henceforth assume that $\calP^{\text{tree}}=\calP^{\text{UC}}$. Under this assumption, the bridges and the tie-lines of the power network $\calG$ also coincide.

\begin{defn}
Given a cascading failure process described by $\calB(n), n\in\set{1,2, \ldots, N}$, the set $\calB(1)$ is said to be its \textbf{\emph{initial failure}}.
\end{defn}

In a power system, it is reasonable to expect that different initial failures can have different levels of impact on the rest of the network. For instance, the disconnection of a single solar panel from the grid is unlikely to cause any disruption to the system operation, while the failure of a transmission line that connects a major generator to the grid may incur significant load shedding. We thus need to distinguish different types of failures and ensure the proposed control scheme reacts accordingly.

\begin{defn} \label{def:critical}
An initial failure $\calB(1)$ is said to be \textbf{\emph{critical}} if the UC optimization \eqref{eqn:uc_olc} is infeasible over $\calG(1):=(\calN, \calE\bs\calB(1))$, or \textbf{\emph{non-critical}} if it is not critical.
\end{defn}

To formally state our localization result, we define the following concept to clarify the precise meaning of a region being ``local'' with respect to an initial failure.

\begin{defn}
Given an initial failure $\calB(1)$, we say that a tree-partition region $\calN_l$ is \textbf{\emph{associated}} with $\calB(1)$ if there exists an edge $e=(i,j)\in\calB(1)$ such that either $i\in\calN_l$ or $j\in\calN_l$.
\end{defn}

As we discuss in Section \ref{section:control_strategy}, our control strategy provides strong guarantees in mitigation and localization for both non-critical and critical failures, in a way that only the operation of the associated regions are adjusted whenever possible.

\section{Proposed Control Strategy}\label{section:control_strategy}
\begin{figure}[t]
\centering
{
\iftoggle{isarxiv}{
\includegraphics[width=0.35\textwidth]{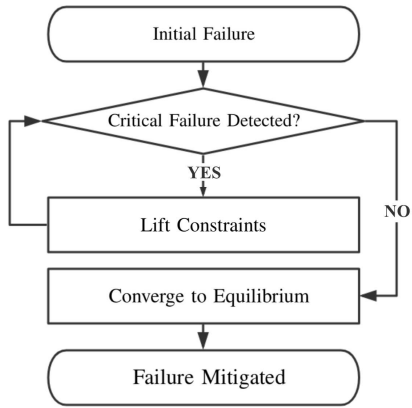}
}{
\includegraphics[width=0.35\textwidth]{figs/flow_chart}
}
}
\caption{Flowchart of the events after an initial failure under the proposed control strategy.}\label{fig:flow_chart}
\end{figure}

Our control strategy revolves around the new and powerful properties of the power system that emerge when the control areas that UC operates over form a tree-partition of the network. In this section, we outline how this strategy can be implemented, in both the planning phase, where a tree-partition structure of the control areas should be created, and the operating phase, during which UC actively monitors and reacts to line failures. Fig.~\ref{fig:flow_chart} illustrates the sequence of events after an initial failure in the proposed control strategy.

\subsection{Planning Phase: Tree-partition of Control Areas}
Power networks are often comprised of multiple control areas, each of which is managed by an independent system operator (ISO). Although these areas exchange power with each other as prescribed by economic dispatch, their operations are relatively independent and it is desirable to ensure system disturbances in one area do not have a significant impact on the others. This is usually achieved via the zero area control error constraint in secondary frequency control \cite{bergen2009power}, and is enforced in UC with \eqref{eqn:uc_ace}. As we discuss in Section \ref{section:connecting_uc_to_tree}, such control areas typically do not form a tree-partition of the transmission network, as having redundant lines is believed to be a crucial part in maintaining $N-1$ security of the power system \cite{bergen2009power,bienstock2007integer,hines2007controlling}.

In order to implement our control strategy, we propose to create a tree-partition whose regions are precisely the control areas over which UC operates. This can be done by switching off a subset of the tie-lines so that the reduced multi-graph obtained from the control area partition forms a tree. The switching actions only need to be carried out in the planning phase as line failures that occur during the operating phase do not affect the tree-partition already in place\footnote{In fact, in certain cases line failures lead to ``finer'' tree partitions as more regions are potentially created when lines are removed from service.}. It is interesting to note that, when the subset of lines to switch off is chosen carefully, this action not only helps localize the impact of line failures, but can also improve the system reliability in the $N-1$ security sense. This seemingly counter-intuitive phenomenon is illustrated by our case studies in Section \ref{section:case_study_n_1}.

\subsection{Operating Phase: Extending the Unified Controller}
Once a tree-partition is formed, the power network under UC operates as a closed-loop system and responds to disturbances such as transmission line failure or loss of generator/load in an automonous manner. In normal conditions where the system disturbances are insignificant, UC always drives the power network back to an equilibrium point that can be interpreted as an optimal solution of \eqref{eqn:uc_olc}. This is the case, for instance, when non-critical failures (see Definition \ref{def:critical}) happen and therefore, such failures are always properly mitigated.

However, in extreme scenarios where a major disturbance (e.g.~a critical failure) affects the system, the optimization problem \eqref{eqn:uc_olc} that UC aims to solve can be infeasible. In other words, it is physically impossible for UC to achieve all of its control objectives after such a disturbance. This causes UC to be unstable (see Proposition \ref{prop:dual_to_infinity}) and, further, leads to successive failures or even large scale outages. As such, there is a need to extend the version of UC proposed in \cite{zhao2016unified, zhao2018distributed} with two features: (a) a critical failure detection component that monitors the system states and ensures UC is aware of such extreme situation promptly when it happens; (b) a constraint lifting component that responds to critical failures by proactively relaxing certain goals that UC tries to achieve, and ensures system stability can be reached at minimal cost. 

Our technical results in Section \ref{section:critical_detection} suggest a way to implement both components as part of the normal operation of UC. System operators can prioritize different control areas by specifying the sequence of constraints to lift in response to extreme events. This allows the non-associated regions to be progressively involved and coordinated in a desired pattern when mitigating critical failures. We present and discuss some potential schemes in Section \ref{section:critical_contraint_lift}.
\subsection{Guaranteed Mitigation and Localization}
As we show in detail in Sections \ref{section:non_critical} and \ref{section:critical}, our control strategy provides strong guarantees in mitigation and localization for both non-critical and critical failures. More specifically, the proposed control strategy ensures that, (a) non-critical failures are always fully mitigated by the associated regions, and the operating points for non-associated regions are not impacted at all; (b) critical failures are guaranteed to be mitigated with certain constraints in \eqref{eqn:uc_olc} being lifted, in a progressive manner specified by the system operator. Thus the proposed strategy always prevents successive failures from happening, while localizing the impact of the initial failures as much as possible.

\section{Localizing Non-critical Failures} 
\label{section:non_critical}
In this section, we consider non-critical failures, as defined in Section \ref{section:basic}, and prove that such failures are always fully mitigated within the associated regions.

We first characterize how the system operating point shifts in response to such failures. Recall that if an initial failure $\calB(1)$ is non-critical, the UC optimization \eqref{eqn:uc_olc} is feasible and thus the new system operating point $x^*(1):=(\omega^*(1), d^*(1), f^*(1), \theta^*(1))$\footnote{We add $\theta$ to the state space of \eqref{eqn:swing_and_network_dynamics} when the phase angle is relevant.} under UC control satisfies all the constraints in \eqref{eqn:uc_olc}. In particular, none of the line limits in \eqref{eqn:line_limit} is violated at $x^*(1)$, i.e.~$x^*(1)$ is a secure operating point and the cascade stops, namely $\calF(1) = \emptyset$.


\begin{lemma}\label{lemma:zero_branch_deviation}
Given a non-critical initial failure $\calB(1)$, the new operating point $x^*(1)$ prescribed by the UC satisfies $f^*_e(1)=0$ for every bridge $e$.
\end{lemma}

The above lemma shows that, in addition to the zero area control error constraints enforced by \eqref{eqn:uc_ace}, when the control areas that UC operates over form a tree-partition, UC further guarantees zero flow deviations on all tie-lines. This demonstrates how a tree-partition enables UC to achieve a stronger performance guarantee compared to its original form as proposed in \cite{zhao2016unified,zhao2018distributed}. The following proposition is another result of this type, which clarifies how the tree-partition brings localization properties to UC.

\begin{prop}\label{prop:localizability}
Assume $c_j(\cdot)$ is strictly convex for all $j\in\calN$. Given a non-critical initial failure $\calB(1)$, if a tree-partition region $\calN_l$ is not associated with $\calB(1)$, then $d^*_j(1)=0$ for all $j\in\calN_l$.
\end{prop}

The core idea underlying the proof of this proposition is easy to explain: Lemma \ref{lemma:zero_branch_deviation} implies the tie-line flows, which are the only coupling among the regions, are zero; thus the UC optimization \eqref{eqn:uc_olc} over different regions are totally ``separated'' and hence, the operating points for non-associated regions should remain unchanged. A rigorous proof is, however, more involved and requires a technical result that relates the solution space of $CBC^T$ to tree-partitions.

\begin{lemma}\label{lemma:laplace_solvability}
Let $\calP^{\text{\emph{tree}}}=\set{\calN_1,\calN_2,\cdots, \calN_l}$ be a tree-partition of $\calG$ and consider a vector $b\in\R^{\abs{\calN}}$ such that $b_j=0$ for all $j\in\calN_1$ and $\sum_{j\in\calN_k} b_j=0$ for $k\neq 1$. Set
$$
\partial\calN_1:=\set{j: j\notin\calN_1, \exists i\in\calN_1 \text{ s.t. } (i,j)\in \calE \text{ or } (j,i)\in\calE}
$$
and $\ol{\calN}_1=\calN_1\cup \partial\calN_1$. Then the linear system
\begin{equation}\label{eqn:laplace_eqn}
CBC^Tx = b
\end{equation}
is solvable, and any solution $x$ to \eqref{eqn:laplace_eqn} satisfies $x_i=x_j$ for all $i,j\in\ol{\calN}_1$.
\end{lemma}
The set $\partial\calN_1$ defined above are the ``boundary'' buses of $\calN_1$ in $\calG$ and $\ol{\calN}_1$ can be interpreted as the closure of $\calN_1$. It has a simple interpretation in the DC power flow context. Think of $b$ as bus injections and $x$ as the phase angles. Suppose the injection at every node in $\calN_1$ is zero and the injections within every other region $\calN_k$ are balanced (i.e., sum to zero). Then Lemma \ref{lemma:laplace_solvability} says that the phase angles are the same at every node in $\ol{\calN}_1$, i.e., the angle difference across every line in or incident to $\calN_1$ is zero. This result only holds if the underlying regions form a tree-partition and its proof is presented in\iftoggle{isarxiv}{
Appendix \ref{section:proof_of_laplace_solvability}}{\cite{report}}. 

\begin{proof}[Proof sketch of Proposition \ref{prop:localizability}]
For the purpose of simplified notations, we drop the stage index $(1)$ from $x^*$ and denote $x^*=(\omega^*,d^*,f^*,\theta^*)$.\iftoggle{isarxiv}{
To streamline the presentation, we only sketch the main ideas of the proof here and leave the details to Appendix \ref{section:proof_of_localizability}.
}{Due to the space limitation, we only sketch the main ideas of the proof here and leave the details to our report \cite{report}.}

First, we construct a different point $\tilde{x}^*$ from $x^*$ as follows: (i) replace $d^*_j$ with $0$ for all $j\in\calN_l$; (ii) replace $f^*_e$ with $0$ for $e\in\calE$ that have both endpoints in $\calN_l$; and (iii) replace $\theta^*$ by a solution 
$\tilde{\theta}^*$ obtained from solving DC power flow equations with injections specified by $\tilde{d}$. Since $c_j(\cdot)$ attains its minimum at $0$, $\tilde{x}^*$ achieves at least the same objective value \eqref{eqn:uc_obj} as $x^*$. Thus $\tilde{x}^*$ must be an optimal point of \eqref{eqn:uc_olc}, provided it is feasible.

Second, as the core step in the whole proof, we apply Lemma \ref{lemma:laplace_solvability} to all regions of $\calP^{\text{tree}}$ separately, and show that $\tilde{\theta}^*$ is consistent with the injections and branch flows specified by $\tilde{x}^*$. This together with routine checks allows us to prove the feasibility of the point $\tilde{x}^*$.

Finally, when the cost functions $c_j(\cdot)$ are strictly convex, the optimal solution to  \eqref{eqn:uc_olc} is unique in $d^*$ and $f*$ ($\theta^*$ is also unique up to a constant shift). We thus conclude that $\tilde{x}^*=x^*$ (up to a constant shift on $\theta$). This completes the proof.
\end{proof}

This result reveals that, with the proposed control strategy, when the system converges to equilibrium after a non-critical failure, the injections and power flows in the non-associated regions remain unchanged. In other words, our control scheme guarantees that non-critical failures in a control area do not impact the operations of other areas at all, achieving a stronger control area independence than that ensured by the zero control error requirement.

Unlike the scheme in \cite{guo2018failure}, bridge failures in the proposed control strategy are treated in exactly the same way as other lines, provided that they are non-critical. Furthermore, the impact of such bridge failures is localized to the associated regions. This contrast with the global impact of bridge failures in \cite{guo2018failure} demonstrates again the benefits of connecting UC to tree-partitions.


\section{Controlling Critical Failures}\label{section:critical}
We now consider the case where the initial failure is critical. This may happen when a major generator or transmission line is disconnected from the grid.
\subsection{Unified Controller under Critical Failures}\label{section:critical_detection}
Since UC is a concept that emerged from the frequency regulation literature, the underlying optimization \eqref{eqn:uc_olc} is always assumed to be feasible in existing studies \cite{zhao2016unified,zhao2018distributed}. As such, little is known about the behaviors of UC if this assumption is violated, which  is the case when a critical failure happens. We now derive a result that closes this gap and characterizes the limiting behavior of UC in this setting.

In order to do so, we first need to formulate the exact controller dynamics of UC. Unfortunately, there is no standard way to do so as multiple designs of UC have been proposed in the literature \cite{zhao2014design,mallada2017optimal,zhao2016unified, zhao2018distributed}, each with its own strengths and weaknesses. Nevertheless, all of the proposed controller design are (approximately) projected primal-dual algorithms to solve the underlying optimization \eqref{eqn:uc_olc}, and satisfy the following assumptions:
\vspace{.1cm}
\newline
\textbf{UC1}: For all $j\in\calN$, $\ul{d}_j\le d_j(t)\le \ol{d}_j$ is satisifed for all $t$. This is achieved either via a projection operator that maps $d_j(t)$ to this interval, or by requiring the cost function $c_j(\cdot)$ to approach infinity near these boundaries.
\newline
\textbf{UC2}: Dual variables are introduced for constraints \eqref{eqn:uc_balance}-\eqref{eqn:line_limit} and maintained throughout the operation (denote these dual variables by $\lambda_i$ for $i\in\set{1,2,\cdots, \abs{\calN} + 3\abs{\calE}+\abs{\calP^{\text{UC}}}}$).
\newline
\textbf{UC3}: The primal variables $f, \theta$ and the dual variables $\lambda_i$ are updated by a primal-dual algorithm\footnote{We do not consider the specific variants of the standard primal-dual algorithms that are proposed in different designs of UC, since the standard primal-dual algorithm is often a good approximation.} to solve \eqref{eqn:uc_olc}.

\begin{prop} \label{prop:dual_to_infinity}
Assume UC1-UC3 hold. If \eqref{eqn:uc_olc} is infeasible, then there exists a dual variable $\lambda_i$ such that:
$$
\limsup_{t\goesto\infty}\abs{\lambda_i(t)} = \infty
$$
\end{prop}
This result implies that after a critical failure, UC cannot drive the system to a proper and safe operating point. In fact, it always leads to instability in the system (certain dual variables can take arbitrarily large values). This drawback, however, when viewed from a different perspective, suggests a way to detect critical failures. More specifically, since Proposition \ref{prop:dual_to_infinity} guarantees certain dual variables will become arbitrarily large in UC operation when \eqref{eqn:uc_olc} is infeasible, we can always set a threshold for the dual variables and raise an infeasibility warning if some of them exceed the corresponding thresholds. By doing so, critical failures can always be detected, and this happens in a distributed fashion in parallel to the normal operation of UC. Moreover, by setting tighter thresholds around the normal operating point, such failures can be detected more promptly.

Of course, this method is subject to false alarms since non-critical failures may also cause relatively large dual variable values in transient state. There is an intrinsic tradeoff on the level of the thresholds to be applied, in the following sense: A tighter threshold allows critical failures to be detected more promptly, yet also leads to a larger false alarm rate. In practice, these thresholds should be chosen carefully by the operator in accordance to the specific system parameters and application scenarios.

\subsection{Constraint Lifting as a Remedy}\label{section:critical_contraint_lift}
In the event of a critical failure, it is physically impossible for UC to simultaneously achieve all of its control objectives. Our discussion in the last subsection shows that, if UC still operates following its normal dynamics, the system is subject to instability and thus  successive failures. In the worst case, this can lead to large scale outages.

We can prevent this from happening by lifting certain constraints from UC. Without compromising the basic objective to stabilize the system, there are two ways to do so:
\begin{itemize}
	\item The zero area control error constraints \eqref{eqn:uc_ace} between certain control areas can be lifted. This in practice means the controller now gets more control areas involved to mitigate the failure.
	\item Certain load shedding can be applied, which in \eqref{eqn:uc_olc} is reflected by enlarging the range $[\ul{d}_j, \ol{d}_j]$ for the corresponding load buses.
\end{itemize}


By iteratively lifting the two types of constraints above, one can guarantee the feasibility of \eqref{eqn:uc_olc} and ensure the system under the proposed control converge to a stable point, which in particular is free from successive failures. This, however, comes with the cost of potential load loss, and thus must be carried out properly. In practice, the iterative relaxation procedure can follow predetermined rules specified by the system operator to prioritize different objectives.

\section{Case Studies}\label{section:case_studies}
\begin{figure}[t]
\centering
\iftoggle{isarxiv}{
\includegraphics[width=0.45\textwidth]{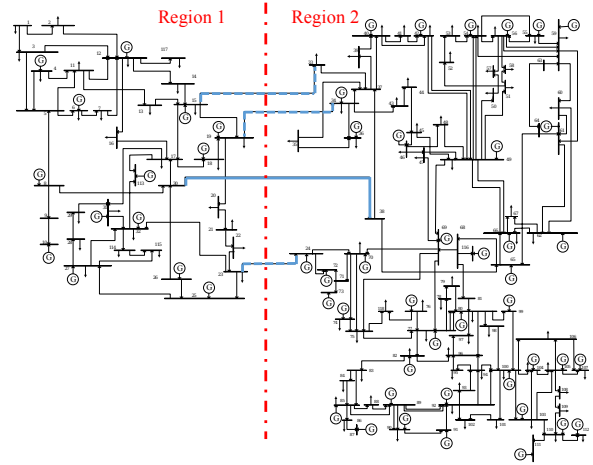}
}{
\includegraphics[width=0.45\textwidth]{figs/IEEE118_2.pdf}}
\caption{One line diagram of the IEEE 118-bus test system with two control areas. Dashed blue lines are switched off when a tree-partition needs to be formed.}
\label{fig:IEEE_118}
\end{figure}
In this section, we evaluate the performance of the proposed control strategy on the IEEE 118-bus test system, which comprises of two control areas as shown in Fig.~\ref{fig:IEEE_118}. The three dashed lines (15, 33), (19, 34) and (23, 24) are switched off whenever a tree-partition needs to be formed, and the new topology is referred to as the revised network. 
The failure scenarios to be examined are created as follows. First, we generate 100 load injections by adding random perturbations (up to 25\% of the base value) to the nominal load profile from \cite{zimmerman2011matpower} and then solve the DC OPF to obtain the corresponding generator operating points. Second, we iterate over every transmission line in the IEEE 118-bus test system as initial failures and simulate the cascading process thus triggered. This produces about 18,000 scenarios.



\subsection{$N-1$ Security}\label{section:case_study_n_1}
We first evaluate the system robustness to failures in terms of the $N-1$ security standard. In particular, we implement both the proposed control strategy and the classical AGC \cite{bergen2009power} on the IEEE 118-bus testbed, and look at the average number of vulnerable lines across all the scenarios that lead to either successive failures or load shedding when they are tripped. In order to illustrate the improvements of the proposed control strategy in different levels of system congestion, we scale down the transmission line capacities to $\alpha=0.9, 0.8, 0.7$ of the base values and collect statistics on the number of vulnerable lines in all these settings.

Our results are summarized in Fig.~\ref{fig:vulnerable_lines}. It can be seen that the proposed control incurs far less number of vulnerable lines in all cases compared to AGC, and this difference is particularly clear when the system is congested. We highlight that this happens with the proposed control operates over the revised network, where some of the tie-lines are switched off and hence certain capacity is removed from the system. Moreover, the remaining tie-line (30, 38) in the revised network is never vulnerable under the proposed control.

%
%

\begin{figure}[t]
\centering
\iftoggle{isarxiv}{
\begin{tikzpicture}
\begin{axis}[
width = 8 cm,
height = 3.5cm,
ymax = 30,
ymin = 0,    
major x tick style = transparent,
enlarge x limits=0.25,
ybar=0pt,
bar width=0.2,
xlabel style = {font = \small},
ylabel style = {font = \small},
yticklabel style = {font = \small},
ylabel = {\# of Vulnerable Lines},
xtick={1.5, 2.5, ..., 4.5},
x tick label as interval,
xticklabels={$\alpha=0.9$,$\alpha=0.8$,$\alpha=0.7$ },
xticklabel style={font = \small, align=center},
legend style={
	at      = {(0.05, 0.95)},
	anchor  = north west,
	font    = \small,
	legend cell align   = left,
	legend columns    = 2,
	transpose legend,
},
]

\addplot+[error bars/.cd,
y dir=both,y explicit]
coordinates {
    (2,4.6) +- (0, 0.77)
    (3, 6.8) +- (0.0, 1.35)
    (4, 16.1) +- (0.0, 8.58)
};
\addplot+[error bars/.cd,
y dir=both,y explicit]
coordinates {
    (2,2.23) +- (0.0, 0.43)
    (3,2.23) +- (0.0, 0.43)
    (4,2.37) +- (0.0, 0.68)
};

\legend{AGC, Proposed Control}
\end{axis}
\end{tikzpicture}
}{
\begin{tikzpicture}
\begin{axis}[
width = 8 cm,
height = 3.5cm,
ymax = 30,
ymin = 0,    
major x tick style = transparent,
enlarge x limits=0.25,
ybar=0pt,
bar width=0.2,
xlabel style = {font = \small},
ylabel style = {font = \small},
yticklabel style = {font = \small},
ylabel = {\# of Vulnerable Lines},
xtick={1.5, 2.5, ..., 4.5},
x tick label as interval,
xticklabels={$\alpha=0.9$,$\alpha=0.8$,$\alpha=0.7$ },
xticklabel style={font = \small, align=center},
legend style={
	at      = {(0.05, 0.95)},
	anchor  = north west,
	font    = \small,
	legend cell align   = left,
	legend columns    = 2,
	transpose legend,
},
]

\addplot+[error bars/.cd,
y dir=both,y explicit]
coordinates {
    (2,4.6) +- (0, 0.77)
    (3, 6.8) +- (0.0, 1.35)
    (4, 16.1) +- (0.0, 8.58)
};
\addplot+[error bars/.cd,
y dir=both,y explicit]
coordinates {
    (2,2.23) +- (0.0, 0.43)
    (3,2.23) +- (0.0, 0.43)
    (4,2.37) +- (0.0, 0.68)
};

\legend{AGC, Proposed Control}
\end{axis}
\end{tikzpicture}
}
\caption{Number of vulnerable lines with respect to different levels of system congestion. }
\label{fig:vulnerable_lines}
\end{figure}
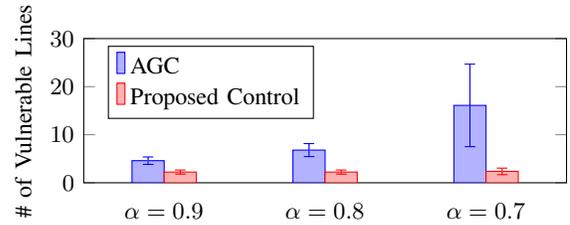


\subsection{Loss of Load and Disruption to System Operation}
We now look at the load loss rate, defined as the ratio between the total loss of load with respect to the original total demand, of the system to evaluate how well failures are mitigated in different settings. In this experiment, we scale down the generator capacities by 35\% and the line capacities by 30\% so that the system is more susceptible to failures.  In order to demonstrate how UC and tree-partition impact the system performance seperately, we look at four different settings: (i) AGC on the original network; (ii) AGC on the revised network; (iii) UC on the original network; and (iv) UC on the revised network. Fig.~\ref{fig:load_loss} plots the complementary cumulative distribution (CCDF) of the load loss rates across all of the failure scenarios in these settings.

As one can see from the figure, for both the original and revised networks, UC significantly outperforms AGC. In particular, the largest load loss rate for UC is less than 2\% for both networks, while AGC can lead to loss rate up to  14\% on the revised network and 21\% on the original network. This demonstrates the  benefits of using our control strategy to mitigate failures.


Although the performance of UC in terms of loss rate are roughly the same with or without tree-partition, there is a drastic difference when we look at how well the failure impacts are localized. In Fig.~\ref{fig:gen_response}, we plot the CCDF on the number of generators whose operating points are adjusted in response to the initial failures. It shows that the operation of much fewer generators is disrupted when the control areas that UC operates over form a tree-partition. This confirms our intuition and theoretical results about how a tree-partition structure helps localize failures.





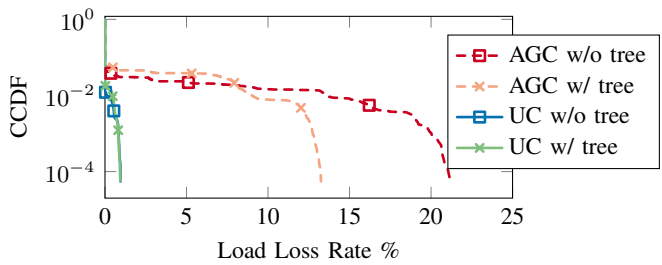
\begin{figure}
\vspace{.1cm}
\centering
\iftoggle{isarxiv}{
\begin{tikzpicture}
\begin{semilogyaxis}[
        ymax    = 1.2,
        ymin    = 0,
        xmax    = 0.25,
        xmin    = 0,
        width   = 7 cm,
        height  = 4 cm,
        xlabel style = {font = \small},
        ylabel style = {font = \small},
        xticklabel style = {font = \small},
        yticklabel style = {font = \small},
    	xlabel = {Load Loss Rate \%},
    	ylabel = {CCDF},
    	xtick = {0,0.05,...,0.3},
    	xticklabels = {0,5,...,30},
    	legend style={
            at      = {(0.84, 0.9)},
            font    = \small,
            anchor  = north west,
            legend cell align   = left,
            legend columns      = 1,
        },
	]
	\pgfplotstableread{PrimaryOriginal.txt}\datatable
	\addplot+[
        mark = none,
        color       = myred,
        line width  = 1pt,
        line cap    = round,
        line join   = round,
        dashed,
        mark = square, 
        mark options={scale=1,solid},
        mark repeat=300,
        mark phase=90,
    ] table {\datatable};
    
	\pgfplotstableread{PrimaryRevised.txt}\datatable
	\addplot +[
        mark = none,
        color       = myorange,
        line width  = 1pt,
        line cap    = round,
        line join   = round,
        dashed,
        mark = x, 
        mark options={scale=1,solid},
        mark repeat=300,
        mark phase=90,
    ] table {\datatable};
    
        	\pgfplotstableread{SecondaryOriginal.txt}\datatable
	\addplot +[
        mark = none,
        color       = myblue,
        line width  = 1pt,
        line cap    = round,
        line join   = round,
        mark = square, 
        mark options={scale=1,solid},
        mark repeat=150,
        mark phase=10,
    ] table {\datatable};
    
	\pgfplotstableread{SecondaryRevised.txt}\datatable
	\addplot +[
        mark = none,
        color       = mygreen,
        line width  = 1pt,
        line cap    = round,
        line join   = round,
        mark = x, 
        mark options={scale=1,solid},
        mark repeat=150,
        mark phase=10,
    ] table {\datatable};
    \legend{AGC w/o tree, AGC w/ tree, UC w/o tree,   UC w/ tree};
\end{semilogyaxis}
\end{tikzpicture}
}{
\begin{tikzpicture}
\begin{semilogyaxis}[
        ymax    = 1.2,
        ymin    = 0,
        xmax    = 0.25,
        xmin    = 0,
        width   = 7 cm,
        height  = 4 cm,
        xlabel style = {font = \small},
        ylabel style = {font = \small},
        xticklabel style = {font = \small},
        yticklabel style = {font = \small},
    	xlabel = {Load Loss Rate \%},
    	ylabel = {CCDF},
    	xtick = {0,0.05,...,0.3},
    	xticklabels = {0,5,...,30},
    	legend style={
            at      = {(0.84, 0.9)},
            font    = \small,
            anchor  = north west,
            legend cell align   = left,
            legend columns      = 1,
        },
	]
	\pgfplotstableread{PrimaryOriginal.txt}\datatable
	\addplot+[
        mark = none,
        color       = myred,
        line width  = 1pt,
        line cap    = round,
        line join   = round,
        dashed,
        mark = square, 
        mark options={scale=1,solid},
        mark repeat=300,
        mark phase=90,
    ] table {\datatable};
    
	\pgfplotstableread{PrimaryRevised.txt}\datatable
	\addplot +[
        mark = none,
        color       = myorange,
        line width  = 1pt,
        line cap    = round,
        line join   = round,
        dashed,
        mark = x, 
        mark options={scale=1,solid},
        mark repeat=300,
        mark phase=90,
    ] table {\datatable};
    
        	\pgfplotstableread{SecondaryOriginal.txt}\datatable
	\addplot +[
        mark = none,
        color       = myblue,
        line width  = 1pt,
        line cap    = round,
        line join   = round,
        mark = square, 
        mark options={scale=1,solid},
        mark repeat=150,
        mark phase=10,
    ] table {\datatable};
    
	\pgfplotstableread{SecondaryRevised.txt}\datatable
	\addplot +[
        mark = none,
        color       = mygreen,
        line width  = 1pt,
        line cap    = round,
        line join   = round,
        mark = x, 
        mark options={scale=1,solid},
        mark repeat=150,
        mark phase=10,
    ] table {\datatable};
    \legend{AGC w/o tree, AGC w/ tree, UC w/o tree,   UC w/ tree};
\end{semilogyaxis}
\end{tikzpicture}
}
\caption{CCDF for load loss rate.}
\vspace{-.1cm}
\label{fig:load_loss}
\end{figure}

\begin{figure}
\centering
\iftoggle{isarxiv}{
\begin{tikzpicture}
\begin{axis}[
        ymax    = 0.4,
        ymin    = 0,
        xmax    = 60,
        xmin    = 0,
        width   = 8 cm,
        height  = 3.5 cm,
        xlabel style = {font = \small},
        ylabel style = {font = \small},
        xticklabel style = {font = \small},
        yticklabel style = {font = \small},
    	xlabel = {\# of Adjusted Generators},
    	ylabel = {CCDF},
    	xtick = {0,15,..., 60},
    	legend style={
            at      = {(0.98, 0.5)},
            font    = \small,
            fill    = none,
            anchor  = south east,
            legend cell align   = left,
            legend columns      = 1,
        },
	]
    
	\pgfplotstableread{GenResponseOriginal.txt}\datatable
	\addplot +[
	mark = none,
        color       = myblue,
        line width  = 1pt,
        line cap    = round,
        line join   = round,
        mark = square, 
        mark options={scale=1,solid},
        mark repeat=15,
        mark phase=3,
    ] table {\datatable};
    
	\pgfplotstableread{GenResponseRevised.txt}\datatable
	\addplot +[
	mark = none,
        color       = mygreen,
        line width  = 1pt,
        line cap    = round,
        line join   = round,
        mark = x, 
        mark options={scale=1,solid},
        mark repeat=15,
        mark phase=3,
    ] table {\datatable};
    \legend{UC w/o tree, UC w/ tree};
    \addplot + [mark=none, color = myviolet, dashed, line width = 1pt] coordinates {(15,0) (15,0.4)};
\end{axis}
\end{tikzpicture}
}{
\begin{tikzpicture}
\begin{axis}[
        ymax    = 0.4,
        ymin    = 0,
        xmax    = 60,
        xmin    = 0,
        width   = 8 cm,
        height  = 3.5 cm,
        xlabel style = {font = \small},
        ylabel style = {font = \small},
        xticklabel style = {font = \small},
        yticklabel style = {font = \small},
    	xlabel = {\# of Adjusted Generators},
    	ylabel = {CCDF},
    	xtick = {0,15,..., 60},
    	legend style={
            at      = {(0.98, 0.5)},
            font    = \small,
            fill    = none,
            anchor  = south east,
            legend cell align   = left,
            legend columns      = 1,
        },
	]
    
	\pgfplotstableread{GenResponseOriginal.txt}\datatable
	\addplot +[
	mark = none,
        color       = myblue,
        line width  = 1pt,
        line cap    = round,
        line join   = round,
        mark = square, 
        mark options={scale=1,solid},
        mark repeat=15,
        mark phase=3,
    ] table {\datatable};
    
	\pgfplotstableread{GenResponseRevised.txt}\datatable
	\addplot +[
	mark = none,
        color       = mygreen,
        line width  = 1pt,
        line cap    = round,
        line join   = round,
        mark = x, 
        mark options={scale=1,solid},
        mark repeat=15,
        mark phase=3,
    ] table {\datatable};
    \legend{UC w/o tree, UC w/ tree};
    \addplot + [mark=none, color = myviolet, dashed, line width = 1pt] coordinates {(15,0) (15,0.4)};
\end{axis}
\end{tikzpicture}
}
\caption{CCDF for generator response.}
\label{fig:gen_response}
\end{figure}
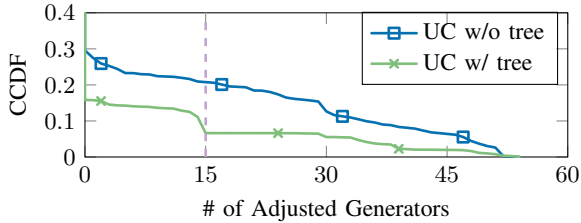

\section{Conclusion}\label{section:conclusion}
In this work, we propose a control strategy that combines the concepts of the unified controller and the network tree-partition to mitigate and localize cascading failures in power system. Our case studies on the IEEE 118-bus test system show that the proposed control scheme greatly improves system robustness to cascading failures as compared to classical AGC. 

This work can be extended in several directions. First, our model builds upon linearized swing and power flow dynamics, which are accurate for small deviations but less so under large disruptions. It is of interest to understand how the non-linearity in more sophisticated models on large deviations impacts our results. Second, the proposed control strategy requires certain tie-lines to be switched off so that a tree-partition is formed. It would be useful if the selection of such lines can be systematically optimized under a certain objective function. Third, both the power flow redistribution and the line capacities are relevant in the cascading failure dynamics. It is important to investigate how adjustments on line capacities can be incorporated to our framework to further improve the system reliability.

\bibliographystyle{IEEEtran}
\bibliography{biblio}

\iftoggle{isarxiv}{
\appendices
\section{Proof of Lemma \ref{lemma:zero_branch_deviation}}\label{section:proof_of_zero_branch_deviation}
Given a bridge $e=(j_1,j_2)$ of $\calG$, removing $e$ from $\calG$ partitions $\calG$ into two connected components, say $\calC_1$ and $\calC_2$. Without loss of generality, assume $j_1\in\calC_1$ and $j_2\in\calC_2$. For a region $\calN_v$ from $\calP$, we say $\calN_v$ is within $\calC_1$ if for any $j\in\calN_v$ we have $j\in\calC_1$. It is easy to check from the definition of tree-partitions that any region $\calN_v$ from $\calP$ is either within $\calC_1$ or within $\calC_2$, and $e$ is the only edge in $\calG$ that has one endpoint in $\calC_1$ and the other endpoint in $\calC_2$.

Let $\calP'$ be the set of regions within $\calC_1$ from $\calP$, and put $\bff{1}_{\calP'}\in\set{0,1}^{\abs{\calP}}$ to be its characteristic vector (that is, the $l$-th component of $\bff{1}_{\calP'}$ is $1$ if $\calN_l\in\calP'$ and $0$ othersize). Given two buses $i$ and $j$, we denote $i\goesto j$ if $(i,j)\in\calE$ and $j\goesto i$ if $(j,i)\in\calE$. With such notations, from \eqref{eqn:uc_ace}, we have
\begin{IEEEeqnarray}{rCl}
0&=&\bff{1}_{\calP'}^TECf^*\nonumber\\
 &=&\sum_{l:\calN_l\in\calP'}\sum_{i\in\calN_l}\paren{\sum_{j:j\goesto i}f^*_{ji} - \sum_{j:i\goesto j} f^*_{ij}}\nonumber\\
 &=&\sum_{i:i\in\calC_1}\paren{\sum_{j:j\goesto i}f^*_{ji} - \sum_{j:i\goesto j} f^*_{ij}}\nonumber\\
 &=&f^*_e+\sum_{i:i\in\calC_1}\paren{\sum_{j:j\goesto i,j\in\calC_1}f^*_{ji} - \sum_{j:i\goesto j, j\in\calC_1} f^*_{ij}},
  \label{eqn:uc_seperation_last_step}
\end{IEEEeqnarray}
where \eqref{eqn:uc_seperation_last_step} is because the only edge with one endpoint in $\calC_1$ and the other endpoint in $\calC_2$ is $e$. Note that
\begin{IEEEeqnarray*}{rCl}
&&\sum_{i:i\in\calC_1}\paren{\sum_{j:j\goesto i,j\in\calC_1}f^*_{ji} - \sum_{j:i\goesto j, j\in\calC_1} f^*_{ij}}\\
&=&\sum_{(i,j)\in\calE_1}\paren{f^*_{ij} - f^*_{ij}}\\
&=&0,
\end{IEEEeqnarray*}
where $\calE_1$ is the set of edges with both endpoints in $\calC_1$. From \eqref{eqn:uc_seperation_last_step}, we see that $f^*_e=0$.

Since the bridge $e$ is arbitrary, we have thus proved the desired result.

\qed
\section{Proof of Proposition \ref{prop:localizability}}\label{section:proof_of_localizability}
We now prove the core step as mentioned in the main body of the paper. To simplify the notations, we drop the stage index $(1)$ from $x^*$ and denote $x^*=(\omega^*,d^*,f^*,\theta^*)$. Put $\tilde{x}^*=(\tilde{d}^*, \tilde{f}^*, \tilde{\theta}^*)$. From the way that $\tilde{x}^*$ is constructed, the constraints \eqref{eqn:uc_ace} are easily seen to be satisfied. If we can show that $\tilde{f}^*=BC^T\tilde{\theta}^*$, then since $\tilde{\theta}^*$ is obtained by solving the DC power flow equations from $CBC^T\tilde{\theta}^*=r-\tilde{d}^*$, the constraints \eqref{eqn:uc_balance} and \eqref{eqn:dcflow} are also satisfied. Now we show that $\tilde{f}^*=BC^T\tilde{\theta}^*$ indeed holds. 

To do so, we first establish the following lemma:
\begin{lemma}\label{lemma:uc_injection_vanish}
For any tree-partition region $\calN_z$ in $\calP$, we have
$$
\sum_{j\in\calN_z}\paren{r_j-d^*_j}=\sum_{j\in\calN_z}\paren{r_j-\tilde{d}^*_j}=0.
$$
\end{lemma}
\begin{proof}
Let $\bff{1}_{\calN_z}\in\R^{\abs{\calN}}$ be the characteristic vector of $\calN_z$, that is, the $j$-th component of $\bff{1}_{\calN_z}$ is $1$ if $j\in\calN_z$ and $0$ otherwise. Summing \eqref{eqn:uc_balance} over $j\in\calN_z$, we have:
\begin{IEEEeqnarray*}{rCl}
\sum_{j\in\calN_z}\paren{r_j-d^*_j}&=&\bff{1}_{\calN_z}^TCf=(ECf)_z=0,
\end{IEEEeqnarray*} 
where $(ECf)_z$ is the $z$-th row of $ECf$.

For $\calN_z$ that is different from $\calN_l$, we have $\tilde{d}^*_j=d^*_j$ for any $j\in\calN_z$ by construction. Thus for such $\calN_z$ we also have
$$
\sum_{j\in\calN_z}\paren{r_j-\tilde{d}^*_j}=0.
$$

For $\calN_l$, since $\calN_l$ is not associated with $\calB(1)$, we have $r_j=0$ for $j\in\calN_l$. Moreover, by construction we also know that $\tilde{d}^*_j=0$ for $j\in\calN_l$. As a result
$$
\sum_{j\in\calN_l}\paren{r_j-\tilde{d}^*_j}=0.
$$

This completes the proof.
\end{proof}

Now consider a region $\calN_w$ that is different from $\calN_l$. In this case, we do not change the injection from $x^*$ when constructing $\tilde{x}^*$, thus $d_j^*-\tilde{d}_j^*=0$ for all $j\in\calN_w$. From Lemma \ref{lemma:uc_injection_vanish}, we see that $\sum_{j\in\calN_z}\paren{d_j^*-\tilde{d}_j^*}=0$ for all $z$. Since $d^*$ and $\theta^*$ conform to the DC power flow equations, we have
$$
CBC^T\theta^*=r-d^*
$$
and thus
$$
CBC^T\paren{\theta^* - \tilde{\theta}^*}=\tilde{d}^* - d^*.
$$
By Lemma \ref{lemma:laplace_solvability}, we then have $\theta^*_j-\tilde{\theta}^*_j$ is a constant over $\ol{\calN}_w$, and thus 
$$
\tilde{\theta}^*_i - \tilde{\theta}^*_j=\theta^*_i -\theta^*_j 
$$
for all $i, j\in\ol{\calN}_w$. This in particular implies
$$
\tilde{f}^*_e = f^*_e =B_e(\theta^*_i -\theta^*_j)= B_e(\tilde{\theta}^*_i - \tilde{\theta}^*_j)
$$
for all $e=(i,j)$ such that $i\in\calN_w$ or $j\in\calN_w$.

Next let us consider the region $\calN_l$. In this region, we have $\tilde{d}^*_j=0$ by construction. Moreover, since $\calN_l$ is not associated with $\calB(1)$, we know $r_j=0$ for all $j\in\calN_l$. Thus $r_j-\tilde{d}^*_j=0$ for all $j\in\calN_l$. Further, from Lemma \ref{lemma:uc_injection_vanish} we have $\sum_{j\in\calN_z}\paren{r_j-\tilde{d}_j^*}=0$ for all $z$. Thus by Lemma \ref{lemma:laplace_solvability} and $CBC^T\tilde{\theta}^*=r-\tilde{d}^*$, we know $\tilde{\theta}^*_i=\tilde{\theta}^*_j$ for all $i,j\in\ol{\calN}_l$. This implies that for any edge $e=(i,j)$ within $\calN_l$, we have
$$
\tilde{f}^*_e=0=B_e(\tilde{\theta}^*_i-\tilde{\theta}^*_j).
$$

As a result, we see that $\tilde{f}^*_e=B_e(\tilde{\theta}^*_i-\tilde{\theta}^*_j)$ holds for all $e\in\calE$. This completes the proof. \qed
\section{Proof of Lemma \ref{lemma:laplace_solvability}}\label{section:proof_of_laplace_solvability}
It is well-known that the Laplacian matrix $L:=CBC^T$ of a connected graph $\calG=(\calN,\calE)$ has rank $\abs{\calN} - 1$, and $Lx=b$ is solvable if and only if $\bff{1}^Tb=0$, where $\bff{1}$ is the vector with a proper dimension that consists of ones. Moreover, the kernel of $L$ is given by $\spann(\bff{1})$.

If $\calN_1$ is the only region in $\calP$, then $b=0$ since $b_j=0$ for all $j\in\calN_1$. We thus know the solution space to $Lx=b$ is exactly the kernel of $L$, and the desired result holds.

If $\calN_1$ is not the only region in $\calP$, we can then find a bus that does not belong to $\calN_1$, say bus $z$. 
Without loss of generality, assume the bus $z\in\calN_k$.
Consider a solution $x$ to $Lx=b$. Since the kernel of $L$ is $\spann(\bff{1})$,  we can without loss of generality assume the last component of $x$ is $0$. Let $\ol{L}$ be the submatrix of $L$ obtained by removing its last row and last column, and similarly let $\ol{x}$ and $\ol{b}$ be the vectors obtained by removing the last component of $x$ and $b$, respectively. Then $\ol{L}$ is invertible (since it is a principal submatrix, see \cite{guo2017monotonicity}), and we have
$$
\ol{L}\ol{x}=\ol{b}.
$$

Denote the matrix obtained by deleting the $l$-th row and $i$-th column of $\ol{L}$ by $\ol{L}^{li}$, then Proposition V.2 of \cite{guo2017monotonicity} shows that
\begin{equation}\label{eqn:graph_interpretation}
\det\paren{\ol{L}^{li}}=(-1)^{l+i}\sum_{E\in\calT(\set{l,i}, \set{z})}\chi(E),
\end{equation}
where $\chi(E)=\prod_{e\in E} B_e$ and $\calT(\set{l,i}, \set{z})$ is the set of spanning forests of $\calG$ that consists of exactly two trees containing $\set{l,i}$ and $\set{z}$ respectively. We refer the readers to \cite{guo2017monotonicity} for a detailed discussion on how to interpret these notations.

To state some useful results derived from  \eqref{eqn:graph_interpretation}, we introduce the following definition of directly connected regions.
\begin{defn}
For a tree-partition $\calP=\set{\calN_1,\calN_2,\cdots, \calN_k}$ of $\calG$, we say $\calN_v$ and $\calN_w$ are \emph{\textbf{directly connected without $\calN_l$}} if the path from $\calN_v$ to $\calN_w$ in $\calG_\calP$ does not contain $\calN_l$. 
\end{defn}
The path from $\calN_v$ to $\calN_w$ in the above definition is unique since $\calG_\calP$ forms a tree. As an example, in Fig.~\ref{fig:tree_partition}, $\calN_1$ and $\calN_2$ are directly connected without $\calN_3$, yet $\calN_2$ and $\calN_3$ are not directly connected without $\calN_1$.

In the following proofs, we need to refer to paths in both the original graph $\calG$ and the reduced graph $\calG_\calP$. To clear potential confusions, we agree the following terminologies: Given two sets of nodes $\calN_v$ and $\calN_w$ (that can be different from the tree-partition regions in $\calP$) of $\calG$, a path in $\calG$ from $\calN_v$ to $\calN_w$ refers to a path consisting of nodes (and lines) from the original graph $\calG$ whose starting node belongs to $\calN_v$ and ending node belongs to $\calN_w$. Given two tree-partition regions $\calN_v$ and $\calN_w$, a path in $\calG_\calP$ from $\calN_v$ to $\calN_w$ refers to a path consisting of nodes (and lines) from the reduced graph $\calG_\calP$ whose starting node is $\calN_v$ and ending node is $\calN_w$. Since there is a natural correspondance between bridges in $\calG$ and lines in $\calG_\calP$, if a line $e$ in $\calG_\calP$ is contained in a path $P$ in $\calG_\calP$, we also say the corresponding bridge $\tilde{e}$ from $\calG$ is contained in $P$.

\begin{lemma}\label{lemma:not_directly_connected}
Assume $\calN_2$ and $\calN_k$ are not directly connected without $\calN_1$. If $l_1,l_2\in\calN_2$ and $i\in \ol{\calN}_1$, then
$$
\calT(\set{l_1,i}, \set{z}) = \calT(\set{l_2,i},\set{z}).
$$
\end{lemma}
\begin{proof}
The path from $\calN_1$ to $\calN_k$ in $\calG_\calP$ contains a bridge in $\calG$ that incidents to $\calN_1$. Denote this bridge as $\tilde{e}$ and let $w$ be the endpoint of $\tilde{e}$ that is not in $\calN_1$. Then it is easy to check that $w$ is a cut node that any path from $\ol{\calN}_1$ to $\calN_2$ in $\calG$ must contain.

Since $\calN_2$ and $\calN_k$ are not directly connected without $\calN_1$, the path from $\calN_2$ to $\calN_k$ in $\calG_\calP$ passes through $\calN_1$. In other words, any path in $\calG$ from $\calN_2$ to $\calN_k$ must pass through a certain node in $\calN_1$, and thus contains a sub-path in $\calG$ from $\calN_1$ to $\calN_k$. This implies that $w$ is contained in any path in $\calG$ from $\calN_2$ to $\calN_k$.

Note that any tree containing $i\in\ol{\calN}_1$ and $l_1\in\calN_2$ induces a path in $\calG$ from $\ol{\calN}_1$ to $\calN_2$ and thus contains $w$. Further, any tree containing $l_2\in\calN_2$ and $z\in\calN_k$ induces a path from $\calN_2$ to $\calN_k$ in $\calG$, and thus also contains $w$. As a result, these two types of trees always share a common node $w$ and cannot be disjoint:
$$
\calT(\set{l_1,i}, \set{l_2,z})= \emptyset.
$$
Similarly $\calT(\set{l_2,i}, \set{l_1,z})= \emptyset$. Therefore
\begin{IEEEeqnarray*}{rCl}
&&\calT(\set{l_1,i}, \set{z})\\
&=&\calT(\set{l_1, l_2, i}, \set{z})\sqcup\calT(\set{l_1,i}, \set{l_2, z})\\
&=&\calT(\set{l_1, l_2, i}, \set{z})\sqcup\calT(\set{l_2,i}, \set{l_1, z})\\
&=&\calT(\set{l_2,i}, \set{z}),
\end{IEEEeqnarray*}
where $\sqcup$ means disjoint union. The desired result then follows.
\end{proof}
\begin{lemma}\label{lemma:directly_connected}
Assume $\calN_2$ and $\calN_k$ are directly connected without $\calN_1$. If $l\in\calN_2$ and $i_1, i_2 \in\ol{\calN_1}$, then
$$
\calT(\set{l,i_1}, \set{z}) = \calT(\set{l,i_2},\set{z}).
$$
\end{lemma}
\begin{proof}
The path from $\calN_1$ to $\calN_k$ in $\calG_\calP$ (denoted as $P_1$) contains a bridge in $\calG$ that incidents to $\calN_1$. Denote this bridge as $\tilde{e}$ and let $w$ be the endpoint of $\tilde{e}$ that does not belong to $\calN_1$. Then it is easy to check that $w$ is a cut node that any path in $\calG$ from $\ol{\calN}_1$ to $\calN_k$ must pass through.

We claim that if $\calN_2$ and $\calN_k$ are directly conncted without $\calN_1$, then any path from $\ol{\calN}_1$ to $\calN_2$ in $\calG$ must also contain $w$. Indeed, suppose not, then the path from $\calN_1$ to $\calN_2$ in $\calG_\calP$ (denoted as $P_2$) contains a bridge in $\calG$ that incidents to $\calN_1$, and this bridge is different from $\tilde{e}$. If $P_1$ and $P_2$ do not have any common super nodes, then concatenating the two paths induces a path in $\calG_\calP$ from $\calN_2$ to $\calN_k$ that passes through $\calN_1$. In other words, the path from $\calN_2$ to $\calN_k$ in $\calG_\calP$ passes through $\calN_1$, contradicting the assumption that $\calN_2$ and $\calN_k$ are directly connected without $\calN_1$. Therefore, $P_1$ and $P_2$ share a common node, say $\calN_3$. However, $P_1$ and $P_2$ induce two different sub-paths in $\calG_\calP$ from $\calN_1$ to $\calN_3$, contracting the assumption that $\calG_\calP$ forms a tree. We thus have proved the claim.
 
Finally, note that any tree containing $i_1\in\ol{\calN}_1$ and $l\in\calN_2$ induces a path in $\calG$ from $\ol{\calN}_1$ to $\calN_2$ and thus contains $w$. Further, any tree containing $i_2\in\ol{\calN}_1$ and $z\in\calN_k$ induces a path in $\calG$ from $\ol{\calN}_1$ to $\calN_k$ and thus contains $w$. Therefore these two types of trees always share a common node $w$ and cannot be disjoint:
$$
\calT(\set{l,i_1},\set{i_2, z})=\emptyset.
$$
Similarly $\calT(\set{l,i_2},\set{i_1, z})=\emptyset$. As a result,
\begin{IEEEeqnarray*}{rCl}
&&\calT(\set{l,i_1}, \set{z})\\
&=&\calT(\set{l,i_1, i_2}, \set{z})\sqcup\calT(\set{l,i_1}, \set{i_2, z})\\
&=&\calT(\set{l,i_1, i_2}, \set{z})\\
&=&\calT(\set{l,i_1, i_2}, \set{z})\sqcup\calT(\set{l,i_2}, \set{i_1, z})\\
&=&\calT(\set{l,i_2}, \set{z}).
\end{IEEEeqnarray*}
\end{proof}

Now since $b_k=\ol{b}_k=0$ for all $k\in\calN_1$, by Cramer's rule, we have
\begin{equation}\label{eqn:component_inversion}
x_i=\ol{x}_i=\frac{\sum_{l\notin \calN_1}(-1)^{l+i}b_k\det\paren{\ol{L}^{li}}}{\det\paren{\ol{L}}}
\end{equation}
for all $i$.

Let $\calP_1$ be set of the regions in $\calP$ that are directly connected to $\calN_k$ without $\calN_1$ and let $\calP_2$ be the remaining regions. For a region $\calN_l\in\calP_1$, let
$$
\chi(\calN_l):=\sum_{E\in\calT(\set{\tilde{l},i},\set{z})}\chi(E),
$$
where $\tilde{l}$ is an arbitrary bus in $\calN_l$. $\chi(\calN_l)$ is well-defined by Lemma \ref{lemma:not_directly_connected}. This together with the assumption $\sum_{j\in\calN_l}b_j=0$ then implies
\begin{IEEEeqnarray*}{rCl}
&&\sum_{l\in \calN_l}(-1)^{l+i}b_l\det\paren{\ol{L}^{li}}\\
&=&\sum_{l\in \calN_l}b_l\paren{\sum_{E\in\calT(\set{l,i},\set{z})}\chi(E)}=\sum_{l\in \calN_l}b_l\chi(\calN_l)\\
&=&\chi(\calN_l)\sum_{l\in\calN_l}b_l\\
&=&0.
\end{IEEEeqnarray*}

As a result
\begin{IEEEeqnarray*}{rCl}
&&\sum_{l\notin \calN_1}(-1)^{l+i}b_l\det\paren{\ol{L}^{li}}\\
&=&\sum_{\calN_l\in\calP_1}\sum_{l\in\calN_l}(-1)^{l+i}b_l\det\paren{\ol{L}^{li}}\\
&&+\sum_{\calN_l\in\calP_2}\sum_{l\in\calN_l} (-1)^{l+i}b_l\det\paren{\ol{L}^{li}}\\
&=&\sum_{\calN_l\in\calP_2}\sum_{l\in\calN_l} (-1)^{l+i}b_l\det\paren{\ol{L}^{li}}\\
&=&\sum_{\calN_l\in\calP_2}\sum_{l\in\calN_l}b_l\paren{\sum_{E\in\calT(\set{l,i},\set{z})}\chi(E)},
\end{IEEEeqnarray*}
which by Lemma \ref{lemma:directly_connected} takes the same value for all $i\in\ol{\calN}_1$. In other words, the equation \eqref{eqn:component_inversion} takes the same value for all $i\in\ol{\calN}_1$. This completes the proof. \qed

\section{Proof of Proposition \ref{prop:dual_to_infinity}}\label{section:proof_of_blow_up}
Frist, let us put $x=[f, d ,\theta]\in\R^{2\abs{\calN}+\abs{\calE}}$ to collect all the decision variables of the UC optimization \eqref{eqn:uc_olc} and rewrite it to a more generic form:
\begin{subequations}\label{eqn:pd_opt}
\begin{IEEEeqnarray}{ll}
\min_{\ul{d} \le d \le \ol{d}} \quad & c(d) \label{eqn:pd_obj}\\
\hspace{.2cm}\text{s.t.} & Ax \le g \label{eqn:pd_ineq}\\
& Cx = h,\label{eqn:pd_eq}
\end{IEEEeqnarray}
\end{subequations}
where $A, C, g, h$ are matrices (vectors) of proper dimensions from the optimization \eqref{eqn:uc_olc}. Let $\lambda_1,\lambda_2$ be the corresponding dual variables to \eqref{eqn:pd_ineq} and \eqref{eqn:pd_eq} respectively and $\lambda:=[\lambda_1;\lambda_2]$ ($[\cdot;\cdot]$ here means matrix concatenation as a column), we can then write the Lagrangian for \eqref{eqn:pd_opt} as 
$$\calL(x, \lambda) = c(d)  + \lam_1^T (Ax-g) + \lam_2^T (Cx-h).$$
Now by the assumption UC3, we know that:
\begin{subequations}\label{eqn:pd_alg}
\begin{IEEEeqnarray}{rCl}
	\dot{\lam}_1 & = & [Ax-g]_{\lam_1}^+ \label{eqn:pd_lam}\\
	\dot{\lam}_2 & = & Cx-h, \label{eqn:pd_mu}
\end{IEEEeqnarray}
\end{subequations}
where the projection operator $[\cdot]_{\lam_1}^+$ is defined component-wise by
\begin{equation}
	([x]_{\lam_1}^+)_i := \begin{cases}
		x_i & \text{if } x_i > 0 \text{ or } \lam_{1,i} >0 \\
		0 & \text{otherwise.}
	\end{cases}
\end{equation}
Consider two closed convex sets $S_1 = \{x| Ax\le g, Cx = h\}$ and $S_2 = \{x|\ul{d} \le d \le \ol{d}\}$. If the optimization \eqref{eqn:uc_olc} is infeasible, then $S_1\cap S_2=\emptyset$, i.e., the sets $S_1$ and $S_2$ are disjoint. As a result, there exists a hyperplane that separates $S_1$ and $S_2$: $\exists p\in \R^{2\abs{\calN}+\abs{\calE}}, q \in \R$ such that 
$$
p^Tx > q, \forall x \in S_1 \text{ and } p^Tx \le q, \forall x \in S_2.
$$
This then implies the system 
$$
\begin{cases}
Ax\le g\\
Cx=h\\
p^Tx\le q
\end{cases}
$$
is not solvable. By Farkas' Lemma, we can then find $w_1, w_2, w_3$ of proper dimensions such that $w_1 \ge 0, w_3 \ge 0$, $A^T w_1 + C^T w_2 + p w_3 = 0$, and $g^T w_1 + h^T w_2 +  qw_3 =\eps< 0$.

Define $z = [w_1; w_2]$. We then see that under the UC controller, we have for any $t$:
\begin{subequations}\label{eqn:pd_infea}
	\begin{IEEEeqnarray}{rCl}
		z^T \dot{\lambda}(t)&= &w_1^T [Ax(t)-g]_\lam^+ + w_2^T (Cx(t)-h) \nonumber \\
		&\ge& w_1^T [Ax(t)-g]_\lam^+ + w_2^T (Cx(t)-h)\nonumber\\
		&& + w_3(p^T x(t) - q) \label{eqn:primal_in_range}\\
		&\ge & w_1^T (Ax(t)-g) + w_2^T (Cx(t)-h)\nonumber\\
		&& + w_3(p^T x(t) - q) \label{eqn:remove_bracket} \\
		&=&\paren{A^Tw_1  + C^Tw_2 + pw_3}x(t)\nonumber \\
		&& - \paren{w_1^T g + w_2^T h + w_3 q}\nonumber \\
		&=& 0 - \eps \nonumber\\
		&>&0,\nonumber
	\end{IEEEeqnarray}
\end{subequations}
where \eqref{eqn:primal_in_range} comes from $w_3\ge 0$ and the assumption UC1, which ensures $x(t)\in S_2$ and thus $p^Tx(t)-q \le 0$, and \eqref{eqn:remove_bracket} comes from $w_1\ge 0$ and the fact that $[x]_\lambda^+\ge x$ for all $x$ (the inequality is component-wise).

As a result, we see that 
$$
z^T \lambda(t) -z^T\lambda(0)> -\eps t
$$
and thus
$$
\lim_{t\goesto \infty} z^T\lambda(t) = \infty.
$$
Finally, by noting
$$
\lim_{t\goesto \infty} z^T\lambda(t) \le w_1^T\limsup_{t\goesto\infty}\abs{\lambda_1(t)} + \abs{w_2}^T\limsup_{t\goesto\infty}\abs{\lambda_2(t)},
$$
the desired result follows. \qed

}{}

\end{document}